\documentclass[
    a4paper,
    12pt,
    DIV=11,
    headinclude=false,
    footinclude=false,
]{scrartcl}

\setkomafont{disposition}{\normalcolor\bfseries}

\newcommand{\LTadd}[1]{}

\newcommand{\LTinput}[1]{}
\newcommand{\YYCleverefInput}[1]{}

\usepackage[T1]{fontenc}
\usepackage[utf8]{inputenc}
\usepackage[main=english]{babel}
\usepackage{csquotes}
\usepackage{microtype}

\usepackage{mathtools}
\mathtoolsset{centercolon,mathic}
\allowdisplaybreaks[2] 
\usepackage{amssymb}
\usepackage{amsthm}
\usepackage{thmtools}
\usepackage{derivative}
\usepackage{enumitem}
\usepackage{booktabs}
\usepackage{makecell}

\newcommand{\statement}[1]{\paragraph{#1}\pdfbookmark[1]{#1}{#1}} 

\renewcommand{\Delta}{\varDelta}

\renewcommand{\Phi}{\varPhi}

\renewcommand{\Psi}{\varPsi}

\renewcommand{\Lambda}{\varLambda}

\renewcommand{\Gamma}{\varGamma}

\renewcommand{\Omega}{\varOmega}
\renewcommand{\Theta}{\varTheta}

\newcommand{\vlr}{c_\mathrm{LR}}
\let\epsi\varepsilon

\DeclarePairedDelimiter{\intervaloo}{\lparen}{\rparen}

\DeclarePairedDelimiter{\intervalco}{\lbrack}{\rparen}
\DeclarePairedDelimiter{\intervalcc}{\lbrack}{\rbrack}

\DeclarePairedDelimiter{\paren}{\lparen}{\rparen}
\DeclarePairedDelimiterXPP{\pospart}[1]{}{\lbrack}{\rbrack}{_+}{#1}
\DeclarePairedDelimiter{\abs}{\lvert}{\rvert}
\DeclarePairedDelimiter{\norm}{\lVert}{\rVert}

    \newcommand{\VERT}[1]{#1|\mkern-1.5mu#1|\mkern-1.5mu#1|}

    \ExplSyntaxOn
    \makeatletter

    \MT_delim_default_inner_wrappers:n{\tnorm}

    \NewDocumentCommand{\tnorm}{ s o m }{
        \IfBooleanTF{#1}{
        \MT_delim_tnorm_star_wrapper:nnn%
            {\VERT{\bgroup\left}}{#3}{\VERT{\aftergroup\egroup\right}}
        }{
            \IfValueTF{#2}{
                \@nameuse{MT_delim_tnorm_nostarscaled_wrapper:nnn}%
                    {\VERT{\@nameuse {\MH_cs_to_str:N #2 l}}}
                    {#3}
                    {\VERT{\@nameuse {\MH_cs_to_str:N #2 r}}}
            }{
                \MT_delim_tnorm_nostarnonscaled_wrapper:nnn%
                    {\VERT{}}
                    {#3}
                    {\VERT{}}
            }
        }
    }

    \makeatother
    \ExplSyntaxOff

\DeclarePairedDelimiter{\commutator}{\lbrack}{\rbrack}

\DeclarePairedDelimiterX{\innerproduct}[2]{\langle}{\rangle}{#1,#2}
\DeclarePairedDelimiterX{\innerp}[2]{\langle}{\rangle}{#1,#2}

\DeclarePairedDelimiter{\List}{\{}{\}}

\DeclarePairedDelimiter{\floor}{\lfloor}{\rfloor}

\DeclarePairedDelimiterXPP{\dd}[1]{d}{\lparen}{\rparen}{}{#1}
\DeclarePairedDelimiterXPP{\dist}[1]{d}{\lparen}{\rparen}{}{#1}
\DeclarePairedDelimiterXPP{\diam}[1]{\SYMdiam}{\lparen}{\rparen}{}{#1}
\DeclarePairedDelimiterXPP{\Exp}[1]{\exp}{\lparen}{\rparen}{}{#1}

\makeatletter
\let\bBigg@@\bBigg@
\renewcommand{\bBigg@}[2]{{%
  \mathchoice
    {\bBigg@@{#1}{#2}}%
    {\bBigg@@{#1}{#2}}%
    {\big@size=.5\big@size\bBigg@@{#1}{#2}}%
    {\big@size=.3\big@size\bBigg@@{#1}{#2}}}}%
\makeatother

\DeclareDocumentCommand{\trace}{s o e{_} m}{
    \IfValueTF{#3}
        {\Tr_{#3}}
        {\Tr}%
    \IfBooleanTF{#1}
        {\paren*{#4}}
        {
            \IfValueTF{#2}
                {\paren[#2]{#4}}
                {\paren{#4}}%
        }%
}
\DeclareDocumentCommand{\Trace}{s o e{_} m}{
    \IfValueTF{#3}
        {\Tr_{#3}}
        {\Tr}%
    \IfBooleanTF{#1}
        {\paren*{#4}}
        {
            \IfValueTF{#2}
                {\paren[#2]{#4}}
                {\paren{#4}}%
        }%
}

\providecommand\given{}

\newcommand\SetSymbol[1][]{%
    \nonscript\,#1\vert
    \allowbreak
    \nonscript\,
    \mathopen{}}
\DeclarePairedDelimiterX\Set[1]\{\}{%
    \renewcommand\given{%
        \SetSymbol[\delimsize]}
    \nonscript\,
    #1
    \nonscript\,
}

\ExplSyntaxOn
\makeatletter
\cs_generate_variant:Nn \tl_if_eq:nnT {onT}
\tl_new:N \__scale_new_delimsize_tl
\cs_new:Nn \__scale_make_bigger_tl:N {
    \tl_if_eq:onT {#1} {\@MHempty} { \tl_set:Nn \__scale_new_delimsize_tl { big } }
    \tl_if_eq:onT {#1} {}          { \tl_set:Nn \__scale_new_delimsize_tl { big } }
    \tl_if_eq:onT {#1} {\big}      { \tl_set:Nn \__scale_new_delimsize_tl { Big } }
    \tl_if_eq:onT {#1} {\Big}      { \tl_set:Nn \__scale_new_delimsize_tl { bigg } }
    \tl_if_eq:onT {#1} {\bigg}     { \tl_set:Nn \__scale_new_delimsize_tl { Bigg } }
    \tl_if_eq:onT {#1} {\middle}   { \tl_set:Nn \__scale_new_delimsize_tl { middle } }
}
\cs_new:Nn \scale_make_bigger_m:N {
    \__scale_make_bigger_tl:N #1
    \use:c { \tl_use:N \__scale_new_delimsize_tl }
}
\cs_new:Nn \scale_make_bigger_l:N {
    \__scale_make_bigger_tl:N #1
    \tl_if_eq:NnTF \__scale_new_delimsize_tl {middle} {
        \tl_set:Nn \__scale_new_delimsize_tl {left}
    }{
        \tl_put_right:Nn \__scale_new_delimsize_tl {l}
    }
    \use:c { \tl_use:N \__scale_new_delimsize_tl }
}
\cs_new:Nn \scale_make_bigger_r:N {
    \__scale_make_bigger_tl:N #1
    \tl_if_eq:NnTF \__scale_new_delimsize_tl {middle} {
        \tl_set:Nn \__scale_new_delimsize_tl {right}
    }{
        \tl_put_right:Nn \__scale_new_delimsize_tl {r}
    }
    \use:c { \tl_use:N \__scale_new_delimsize_tl }
}
\DeclarePairedDelimiterXPP{\pdd}[1]{\scale_make_bigger_l:N\delimsize\lparen d}{\lparen}{\rparen}{\scale_make_bigger_r:N\delimsize\rparen}{#1}
\DeclarePairedDelimiterXPP{\pdist}[1]{\scale_make_bigger_l:N\delimsize\lparen d}{\lparen}{\rparen}{\scale_make_bigger_r:N\delimsize\rparen}{#1}
\DeclarePairedDelimiterXPP{\pdiam}[1]{\scale_make_bigger_l:N\delimsize\lparen \SYMdiam}{\lparen}{\rparen}{\scale_make_bigger_r:N\delimsize\rparen}{#1}

\makeatother
\ExplSyntaxOff

\usepackage[svgnames]{xcolor}

\usepackage[
    hidelinks,
    bookmarksnumbered,
    bookmarksopen,
    bookmarksopenlevel=2,
    bookmarksdepth=3,
    pdfusetitle=true,
    breaklinks=true,
]{hyperref}

\makeatletter
\@ifpackageloaded{hyperref}{
    \usepackage[capitalise,noabbrev]{cleveref}
    \usepackage{orcidlink}
}{
    \usepackage[capitalise,noabbrev,poorman]{cleveref}
    \newcommand{\texorpdfstring}[2]{#1}
    \newcommand{\href}[2]{#2}
    \newcommand{\hypersetup}[1]{}
    \newcommand{\orcidlink}[1]{ORCiD}
    \newcommand{\pdfbookmark}[1]{}
}
\makeatother

\crefdefaultlabelformat{#2#1#3}
\crefname{equation}{}{}
\crefname{assumption}{Assumption}{Assumptions}
\crefname{conjecture}{Conjecture}{Conjectures}

\usepackage{caption}
\usepackage{subcaption}

\newcommand{\sumstack}[2][]{\ifstrempty{#1}{\sum_{\substack{#2}}}{\smashoperator[#1]{\sum_{\substack{#2}}}}}

\newcommand{\e}{{\mathrm{e}}}

 \newcommand{\R}{\mathbb{R}}
\newcommand{\C}{\mathbb{C}}
\newcommand{\N}{\mathbb{N}}
\newcommand{\Z}{\mathbb{Z}}

\newcommand{\alg}{\mathcal{A}}
\newcommand{\algloc}{\alg_{\mathrm{loc}}}

\newcommand\cexpsym{\mathbb{E}}
\ExplSyntaxOn
\DeclareDocumentCommand{\cexp}{s o m m}{%
    \cexpsym\c_math_subscript_token{#3}
    \IfBlankF{#4}
    {
        \exp_last_unbraced:Ne \paren {\IfBooleanT{#1}{*}\IfValueT{#2}{[\exp_not:N #2]}} {#4}
    }
}
\ExplSyntaxOff

\newcommand{\calI}{\mathcal{I}}

\DeclareMathOperator{\Tr}{Tr}
\newcommand{\SYMdiam}{\operatorname{diam}}

\newcommand{\quadtext}[1]{\quad\text{#1}\quad}

\newcommand{\Alignindent}{\hspace*{2em}&\hspace*{-2em}}

\let\oldsetminus\setminus
\newbox\mybox
\newcommand\cutsetminus[1]{%
    \setbox\mybox\hbox{\(#1\oldsetminus\)}%
    \ht\mybox=0pt%
    \dp\mybox=0pt%
    \usebox\mybox%
}
\renewcommand\setminus{%
    \mathbin{%
        \mathchoice%
            {\displaystyle\oldsetminus}
            {\textstyle\oldsetminus}
            {\cutsetminus{\scriptstyle}}
            {\cutsetminus{\scriptscriptstyle}}
    }%
}

\makeatletter
\ExplSyntaxOn

\newcounter{theoremenv}
\newcounter{theoremenvglobal}
\newlist{thmlist}{enumerate}{1}
\setlist[thmlist]{
    label=\textup{(\alph{thmlisti})},
    ref={(\alph{thmlisti})},
    nosep,
}
\renewcommand{\p@thmlisti}{\perh@ps{\protect\ref{auto-label:\arabic{theoremenvglobal}}}}
\DeclareRobustCommand{\perh@ps}[1]{#1}
\newcommand{\itemref}[1]{%
    \begingroup 
    \let\perh@ps\@gobble\ref{#1}%
    \endgroup
}

\addtotheorempostheadhook{
    \crefalias{thmlisti}{\thmt@envname}
    \setcounter{theoremenv}{\value{\thmt@envname}}
    \stepcounter{theoremenvglobal}
    \exp_args:NNe \renewcommand \thetheoremenv {\use:c{the\thmt@envname}}
    \exp_args:Ne \label {auto-label:\arabic{theoremenvglobal}}
}
\makeatother
\ExplSyntaxOff

\theoremstyle{plain}

\declaretheorem[
    name=Theorem,
]{theorem}

\declaretheorem[
    name=Lemma,
    sibling=theorem,
]{lemma}

\declaretheorem[
    name=Proposition,
    sibling=theorem,
]{proposition}

\theoremstyle{definition}

\declaretheorem[
    name=Definition,
    sibling=theorem,
]{definition}

\declaretheorem[
    name=Assumption,
]{assumption}

\theoremstyle{remark}
\declaretheorem[
    name=Remark,
    sibling=theorem,
    qed=\(\diamond\),
]{remark}

\usepackage{xurl}

\usepackage{dsfont}

\usepackage[
    bibstyle=phys,
    citestyle=numeric,
    biblabel=brackets,
    mincitenames=1,
    maxcitenames=2,
    minalphanames=3,
    maxalphanames=4,
    minbibnames=3,
    maxbibnames=10,
    arxiv=abs,
    eprint=true,
    doi=false,
    url=false,
    sorting=nty,
    useprefix=true,
    alldates=year,
    urldate=long,
]{biblatex}
\renewbibmacro*{note+pages}{%
    \printfield{note}%
    \setunit{\bibpagespunct}%
    \printfield{pages}%
    \newunit%
}
\makeatletter
\DeclareFieldFormat{eprint:arxiv}{%
    \ifhyperref
    {\href{https://arxiv.org/\abx@arxivpath/#1}{%
            arXiv\addcolon\linebreak[2]#1}}
    {arXiv\addcolon
        \nolinkurl{#1}%
    }
}
\makeatother

\renewbibmacro*{maintitle+title}{%
    \iffieldsequal{maintitle}{title}{%
        \clearfield{maintitle}%
        \clearfield{mainsubtitle}%
        \clearfield{maintitleaddon}%
    }{%
        \iffieldundef{maintitle}{}{%
            \printtext[doi/url-link]{%
                \usebibmacro{maintitle}%
            }%
            \newunit\newblock
            \iffieldundef{volume}{}{%
                \printfield{volume}%
                \printfield{part}%
                \setunit{\addcolon\space}%
            }%
        }%
    }%
    \printtext[doi/url-link]{%
        \usebibmacro{title}%
    }%
    \newunit%
}

\DeclareBibliographyDriver{article}{%
    \usebibmacro{bibindex}%
    \usebibmacro{begentry}%
    \usebibmacro{author/translator+others}%
    \setunit{\labelnamepunct}\newblock
    \usebibmacro{title}%
    \newunit
    \printlist{language}%
    \newunit\newblock
    \usebibmacro{byauthor}%
    \newunit\newblock
    \usebibmacro{bytranslator+others}%
    \newunit\newblock
    \printfield{version}%
    \newunit\newblock
    \printtext[doi/url-link]{%
        \usebibmacro{journal+issuetitle}%
        \newunit
        \usebibmacro{byeditor+others}%
        \newunit
        \usebibmacro{note+pages}%
        \newunit\newblock
        \iftoggle{bbx:isbn}
        {\printfield{issn}}
        {}%
    }%
    \newunit\newblock
    \printfield{pubstate}%
    \setunit{\addspace}%
    \printfield{year}%
    \newunit\newblock
    \usebibmacro{doi+eprint+url}%
    \newunit\newblock
    \printfield{addendum}%
    \setunit{\bibpagerefpunct}\newblock
    \usebibmacro{pageref}%
    \newunit\newblock
    \iftoggle{bbx:related}
        {\usebibmacro{related:init}%
        \usebibmacro{related}}
        {}%
    \usebibmacro{finentry}%
}

\newcommand{\mL}{\mathcal{L}}
\newcommand{\mA}{\mathcal{A}}

\newcommand{\eps}{\varepsilon}
\newcommand{\E}{\mathbb{E}}
\renewcommand{\d}{\mathrm{d}}
\renewcommand{\i}{\mathrm{i}}
\newcommand{\w}{\omega}
\newcommand{\tr}{\mathrm{tr}}

\makeatletter
\def\blindfootnote{\gdef\@thefnmark{}\@footnotetext}
\makeatother
\newcommand{\emaillink}[1]{\href{mailto:#1}{#1}}

\title{Dynamics generated by\\ spatially growing derivations on\\ quasi-local algebras}

\date{8. December 2025}

\author{
    Stefan Teufel%
    \texorpdfstring{%
        \,\orcidlink{0000-0003-3296-4261}
    }{}%
    \and
    Marius Wesle%
    \texorpdfstring{%
        \,\orcidlink{0009-0006-3365-4037}
    }{}%
    \and
    Tom Wessel%
    \texorpdfstring{%
        \,\orcidlink{0000-0001-7593-0913}
    }{}%
}

\publishers{%
    \bigskip\foreignlanguage{ngerman}{
        Fachbereich Mathematik,
        Universität~Tübingen,
        \\Auf~der~Morgenstelle~10,
        72076~Tübingen%
    },
    Germany%
}
\setkomafont{publishers}{\normalfont\large}

\newcommand{\relskip}{\mathrel{\phantom{=}}\mathord{}}

\begin{document}

\bgroup
\hypersetup{hidelinks}
\maketitle\thispagestyle{empty}
\blindfootnote{
    \parbox[t]{.8\textwidth}{
        Email:
        \parbox[t]{.5\textwidth}{
            \emaillink{stefan.teufel@uni-tuebingen.de}
            \\\emaillink{marius.wesle@uni-tuebingen.de}
            \\\emaillink{tom.wessel@uni-tuebingen.de}
        }
    }
}
\egroup

\begin{abstract}
    \noindent
    We prove global existence and uniqueness of dynamics on the quasi-local algebra \(\mA\) of a quantum lattice system for spatially growing derivations \(\mL_\Phi = \sum_x \commutator{\Phi_x , \cdot}\).
    Existing results assume that the local terms \(\Phi_x\in\mA\) of the generator are uniformly bounded in space with respect to appropriate weighted norms \(\norm{\Phi_x}_{G,x}\).
    Analogous to the global existence result for first order ODEs, we show that global existence and uniqueness persist if the size of the local terms \(\norm{\Phi_x}_{G,x}\) grows at most linearly in space.
    This considerably enlarges the class of derivations known to have well-defined dynamics.
    Moreover, we obtain Lieb-Robinson bounds with exponential light cones for such dynamics.

    For the proof, we assume Lieb-Robinson bounds with linear light cones for dynamics, whose generators have uniformly bounded local terms.
    Such bounds are known to hold, for example, if the local terms are of finite range or exponentially localized.
\end{abstract}

\section{Introduction}

In this work we consider interactions \(\sum_{x\in\Gamma} \Phi_x\) defined on the CAR algebra \(\mA\) of lattice fermions on some discrete metric space \((\Gamma, d)\) with \(D\)-dimensional volume growth (think of \(\Z^D\) as the standard example), or on the quasi-local algebra of a spin system on \(\Gamma\), for which the local terms \(\Phi_x\in\mA\) are not bounded uniformly, but instead satisfy a linear growth bound of the form
\begin{equation}
    \label{introbound}
    \norm{\Phi_x}_{G,x}
    \leq
    C_\Phi \, \paren[\big]{1 + \dist{x,x_0}}
    \quadtext{for all}
    x \in \Gamma
    .
\end{equation}
Here \(G\) is a fixed decay function, \(\norm{\cdot}_{G,x}\) a weighted norm centred at \(x\) that quantifies the decay around~\(x\), and \(x_0\) a fixed point in \(\Gamma\).
We prove that if \(G\) decays fast enough such that interactions \(\sum_{x\in\Gamma} \Psi_x\) with uniformly bounded \(\norm{\cdot}_{G,x}\)-norms, i.e.\ with
\begin{equation*}
    \tnorm{\Psi}_G \coloneq \sup_{ x} \, \norm{\Psi_x}_{G,x}
    <
    \infty,
\end{equation*}
satisfy a Lieb-Robinson bound with linear light cone and Lieb-Robinson velocity proportional to \(\tnorm{\Psi}_G\), then \(\Phi\) generates a unique one-parameter group of automorphisms of \(\mA\) with exponential light cones.
Previous results in this direction are known to us only for one-dimensional systems, where the existence of the dynamics for linearly growing generators with exponentially decaying terms can also be concluded from~\cite[Theorem~6.2.6]{BR1981} via bounds on the surface energy.

Before going into details, let us briefly sketch the underlying heuristic picture.
For uniformly bounded interactions with \(\tnorm{\Psi}_G < \infty\), Lieb-Robinson bounds control the speed at which the Heisenberg dynamics generated by such an interaction effectively spreads the support of observables uniformly in space.
If, for example, the local terms have uniform finite range or decay exponentially, then the support of any local observable can spread at most with speed \(v_\mathrm{LR} \sim \vlr \, \tnorm{\Psi}_G\), the so called Lieb-Robinson velocity.
Such Lieb-Robinson bounds first proved in finite volume can then be used to prove existence of the dynamics in infinite volume, see for example~\cite{NSY2019} and references therein.
The situation is vaguely analogous to global existence of solutions to first order ODEs on \(\R^D\).
If the velocity field \(v\colon\R^D\to \R^D\) is locally Lipschitz continuous and bounded, unique local solutions extend to unique global solutions, as integral curves can only travel finite distances in finite time.
However, for first order ODEs the local Lipschitz condition, together with a linear upper bound, is sufficient to guarantee global existence, boundedness of the velocity field \(v\) is not needed.
Even when the velocity field grows linearly in space, integral curves can not reach infinity in finite time, instead the distance to the starting point can grow at most exponentially in time.
Our results establish a similar behaviour for the dynamics generated by interactions that satisfy~\eqref{introbound} and may additionally be also time-dependent:
We prove global existence and uniqueness of dynamics and exponential light cones for such interactions, see Theorem~\ref{thm:all-times}.

While we consider our results interesting in their own right because they considerably extend the class of interactions known to generate global dynamics on \(\mA\), let us briefly mention the application that motivated our study of this question: Consider the Hamiltonian \(H^B = \sum_x\Phi^B_x\) of a fermion system subject to a constant magnetic field~\(B\).
While \(H^B\) is typically a bounded interaction, the derivative \(\partial_B H^B\) of \(H^B\) with respect to \(B\) is an interaction with linearly growing local terms.
And \(\partial_B H^B\), or more precisely its image \(\calI(\partial_B H^B)\) under the quasi-local inverse \(\calI\) of the Liouvillian \(\mL_{H^B}\), is expected to generate the spectral flow for gapped ground states of \(H^B\).
So our result is the basis for showing that the spectral flow exists as a cocycle of locally generated automorphisms of \(\mA\) for gapped phases of matter with varying magnetic fields.
We refer to~\cite{WMM+2025} for a short discussion of this problem and to~\cite{MO2020,BTW2025automorphic} for the spectral flow of gapped ground states in infinite volume.

Finally, our result provides a class of automorphisms that are Fréchet continuous but not of Lieb-Robinson type in the sense of~\cite{BNG2025}.
These automorphisms provide interesting examples for the second part of their Theorem~1.1.
However, we provide bounds on the commutator in~\eqref{eq:commutator-LRB}, suggesting that the notion of Lieb-Robinson type in~\cite{BNG2025} might be too restrictive.

Our paper is organized as follows.
Section~\ref{sec:setup} presents the general setup, and Section~\ref{sec:res} states the precise assumptions and results.
The proofs are given in Section~\ref{sec:proofs}.

\section{Mathematical setup}
\label{sec:setup}

In the following we will denote by \((\Gamma, d )\) a countable metric space that is \(D\)-regular, i.e.\ there is a constant \(C_\mathrm{vol}\), such that for all \(x\in \Gamma\) and \(r>0\) we have
\begin{equation*}
    \abs{B_r(x)} \leq C_\mathrm{vol} \, (1+r)^D
    ,
    \quadtext{where}
    B_r(x) \coloneq \Set{y \in \Gamma \given \dist{y,x}\leq r}
\end{equation*}
denotes the closed ball of radius~\(r\).
Standard examples for \(\Gamma\) are \(\Z^D\) or any other Delone set in~\(\R^D\) with the restriction of the Euclidean metric from \(\R^D\).

The antisymmetric (or fermionic) Fock space over \(\Gamma\) with local space \(\C^n\), \(n\in\N\), is
\begin{equation*}
    \mathcal{F}(\Gamma,\C^n)
    \coloneq
    \bigoplus_{N=0}^{\infty}\ell^2(\Gamma,\C^n)^{\wedge N}.
\end{equation*}
We use \(a^*_{x,i}\) and \(a^{}_{x,i}\) for \(x\in \Gamma\), \(i\in \List{1,\dots,n}\), to denote the fermionic creation and annihilation operators associated to the standard basis of \(\ell^2(\Gamma,\C^n)\) and recall that they satisfy the canonical anti-commutation relations (CAR).
The number operator at site \(x\in \Gamma\) is defined by
\begin{equation*}
    n_x \coloneq \sum_{i=1}^n a^*_{x,i} \, a^{}_{x,i}.
\end{equation*}
The algebra of all bounded operators on \(\mathcal{F}(\Gamma,\C^n)\) is denoted by \(\mathcal{B}(\mathcal{F}(\Gamma,\C^n))\).
For each \(M\subseteq \Gamma\) let \(\mA_M\) be the unital C\(^*\)-subalgebra of \(\mathcal{B}(\mathcal{F}(\Gamma,\C^n))\) generated by
\begin{equation*}
    \Set[\big]{a^*_{x,i} \given x\in M,~ i\in \List{1,\dots,n}}.
\end{equation*}
The C\(^*\)-algebra \(\mA \coloneq \mA_{\Gamma}\) is the CAR-algebra, which we also call the \emph{quasi-local} algebra.
We write \(P_0(\Gamma) \coloneq \Set{M\subseteq \Gamma \given \abs{M}<\infty }\) and call
\begin{equation*}
    \mA_{\mathrm{loc}} \coloneq \bigcup_{M\in P_0(\Gamma)} \mA_M \subseteq \mA
\end{equation*}
the \emph{local} algebra, which is dense in \(\mA\).
An operator is called quasi-local if it lies in \(\mA\) and local if it lies in \(\mA_{\mathrm{loc}}\).
There is a unique automorphism%
\footnote{%
    In the following the term \textit{automorphism} is used in the sense of a \(*\)-automorphism as defined for example in~\cite{BR1979}.
}
\(\Theta\) of \(\mA\), such that
\begin{equation*}
    \Theta(a^*_{x,i})
    =
    - \, a^*_{x,i},
    \quadtext{for all}
    x\in \Gamma
    \text{ and }
    i\in \List{ 1,\dots,n }
    .
\end{equation*}
One defines the set of even quasi-local operators
\begin{equation*}
    \mA^+ \coloneq \Set{A\in \mA \given \Theta (A) = A }
    .
\end{equation*}
It is equal to the norm closure of the set of all linear combinations of products of an even number of annihilation or creation operators.
We denote its part in \(M\subseteq \Gamma\) by \(\mA^+_M \coloneq \mA^+\cap \mA_M\).
For disjoint regions \(M_1\),~\(M_2 \subseteq \Gamma\), all operators \(A\in \mA_{M_1}^+\) and \(B\in \mA_{M_2}\) commute, \([A,B] = 0\).

Positive linear functionals of the quasi-local algebra \(\w \colon \mA \to \C\) of norm \(1\) are called states.
In order to define quantitative notions of localization for quasi-local operators, one makes use of the fact that one can localize operators to given regions by means of the fermionic conditional expectation.
To this end first note that \(\mA\) has a unique state \(\w^{\tr}\) that satisfies
\begin{equation*}
    \w^{\tr}(AB) = \w^{\tr}(BA)
\end{equation*}
for all \(A\),~\(B \in \mA\), called the tracial state (e.g.\ \cite[Definition 4.1, Remark 2]{AM2003}).

\begin{proposition}[{\cite[Theorem 4.7]{AM2003}, \cite[Proposition 2.1]{WMM+2025}}]
    \label{Ex+UniqueExpectation}
    For each \(M\subseteq \Gamma\) there exists a unique linear map
    \begin{equation*}
        \E_M\colon \mA \to \mA_M,
    \end{equation*}
    called the \emph{conditional expectation} with respect to \(\w^\tr\), such that
    \begin{equation}
        \label{eq:Conditional-expectation-defining-property}
        \forall A\in \mA
        \ \forall B\in \mA_M:
        \quad
        \w^{\tr}(AB)=\w^{\tr}(\E_M(A)B) .
    \end{equation}
    It is unital, positive and has the properties
    \begin{align*}
        \forall M\subseteq \Gamma
        \ \forall A,C\in \mA_M
        \ \forall B\in\mA
        :&\quad
        \E_M (A\,B\,C) = A\, \E_M(B) \, C
        \\
        \forall M_1,M_2 \subseteq \Gamma
        :&\quad
        \E_{M_1} \circ \E_{M_2} = \E_{M_1\cap M_2}
        \\
        \forall M \subseteq \Gamma
        :&\quad
        \E_M \mA^+ \subseteq \mA^+
        \\
        \forall M \subseteq \Gamma
        \ \forall A \subseteq \mA
        :&\quad
        \norm{\E_M(A)}\leq \norm{A}
        .
    \end{align*}
\end{proposition}

\begin{remark}
    Note that~\cite[Theorem 4.7]{AM2003} discusses only the case of \(\Gamma = \Z^D\).
    The proof however applies in the same way to our setting.
\end{remark}

Note that \(\norm{(1-\E_{B_r(x)})A} \to 0\) as \(r\to\infty\) for all \(A\in \mA\) by density of \(\algloc\) in~\(\mA\).
We now introduce subspaces of \(\mA\) for which one can explicitly control the rate of convergence in this limit in terms of decay functions.

\begin{definition}\label{def:decay-function}
    We call a bounded function \(F\colon [0,\infty) \to (0,\infty)\) a \emph{decay function} and define
    \begin{equation*}
        \nu_F
        \coloneq
        \sup\Set[\big]{ \nu\geq0 \given \sup_{r\geq0} \, F(r)\, (1+r)^\nu < \infty }
        \in
        \intervalco{0,\infty} \cup \List{\infty}
        .
    \end{equation*}
\end{definition}

\begin{definition}
    Let \(F\) be a decay function.
    We say an observable \(A\in \mA\) is \emph{\(F\)-localized} if for all \(x\in \Gamma\) it holds that
    \begin{equation*}
        \norm{A}_{F,x}
        \coloneq
        \norm{A} + \sup_{r\geq0} \frac{\norm{(1-\E_{B_r(x)})A}}{F(r)}
        <
        \infty .
    \end{equation*}
    We denote the space of all \(F\)-localized observables with \(\mA_F\).
    For \(\nu \geq 0\) and \(F(r) \coloneq (1+r)^{-\nu}\) we abbreviate \(\norm{\cdot}_{\nu,x} \coloneq \norm{\cdot}_{F,x}\) and \(\mA_\nu \coloneq \mA_F\).
\end{definition}
We included \(F\equiv 1\) in the class of decay functions, because then the quasi-local algebra \(\mA\) itself appears in the scale of spaces \(\mA_\nu\) at \(\nu=0\).
More precisely, we have
\(\mA_0 = \mA\) and \(\norm{A}_{0,x} \leq 3 \, \norm{A}\).
Also note that for decay functions \(F\) with exponential or slower decay and all \(x_1\),~\(x_2 \in \Gamma\), the norms \(\norm{\cdot}_{F,x_1}\) and \(\norm{\cdot}_{F,x_2}\) are equivalent.
Nevertheless, it is useful to define the family of norms with varying centre in order to express localization of observables more quantitatively~\cite{TW2025note}.

\begin{definition}
    Let \(I \subseteq \R\) be an interval.
    A \emph{time-dependent zero-chain on \(I\)} is a map
    \begin{equation*}
        \Phi \colon I \times \Gamma \to \mA^{+},\ (t,x) \mapsto \Phi_x(t),
    \end{equation*}
    such that for all \((t,x) \in I\times \Gamma\) the operator \(\Phi_x(t)\) is self-adjoint, for each \(x\in \Gamma\), the map \(I \to \mA^{+},\ t \mapsto \Phi_x(t)\) is norm-continuous and for each \(t\in I\) and \(A \in \mA_{\mathrm{loc}}\) the sum
    \begin{equation*}
        \mL_{\Phi(t)} \, A
        \coloneq
        \sum_{x\in \Gamma} \, \commutator{\Phi_x(t) , A}
    \end{equation*}
    converges unconditionally.

    Let \(F\) be a decay function.
    A time-dependent zero-chain \(\Phi\) on an interval \(I\) is \emph{uniformly \(F\)-local} if
    \begin{equation*}
        \tnorm{\Phi}_{F}
        \coloneq
        \sup_{t\in I} \, \sup_{x\in \Gamma} \, \norm{\Phi_x(t)}_{F,x}
        <
        \infty .
    \end{equation*}
    We denote the space of all uniformly \(F\)-local time-dependent zero-chains on \(I\) with~\(\mathcal{Z}_{F,I}\).
\end{definition}

In the analysis of quantum lattice systems, it is more common to specify the generators by so-called interactions, which associate a strictly local operator to each finite set~\(M\subseteq \Gamma\).
While there is no canonical identification of the set of interactions with the set of zero-chains, there are several natural maps that preserve the associated derivation and decay-properties.
For example for each \(x\in \Gamma\) one can sum all terms of an interaction that are centred around \(x\) in a suitable sense to obtain a zero-chain and one can cut each quasi-local term of a zero-chain in a telescopic fashion via the conditional expectation to obtain an interaction.
See, for example, \cite{BTW2025automorphic} for more details on these procedures and references~\cite{KS2020hall,KS2022local} for the motivation behind the term \enquote{zero-chains}.
Also note that the sets of derivations on \(\mA_{\mathrm{loc}}\) obtained from interactions and from zero-chains are exactly the same, namely the antisymmetric \(*\)-derivations from \(\mA_{\mathrm{loc}}\) to \(\mA\) that commute with the parity automorphism \(g_\pi\).
This can be seen as follows: Each such derivation is given by an interaction as is shown in~\cite{AM2003}.
Each interaction has an associated zero-chain with the same derivation, e.g.\ \cite{BTW2025automorphic}.
And from the definition above it is easy to see that every interaction coming from a zero-chain again satisfies the properties mentioned above.
In this work, we use zero-chains, because they allow for a very clear characterization of linearly growing generators in \cref{ass:growth}.

Finally, let us define what it means for a zero-chain to generate a cocycle of automorphisms.

\begin{definition}
    Let \(I\subseteq \R\) be an interval.
    A \emph{cocycle of automorphisms on \(I\)} is a family \((\alpha_{s,t})_{s,t \in I}\) of automorphisms on \(\mA\), such that for all \(s\),~\(t\),~\(u \in I\)
    \begin{equation*}
        \alpha_{s,t} \, \alpha_{t,u}
        =
        \alpha_{s,u}.
    \end{equation*}
    Let \(\Phi\) be a time-dependent zero-chain on \(I\).
    We say the cocycle of automorphisms \((\alpha_{s,t})_{s,t\in I}\) is \emph{generated by \(\Phi\)} if for all \(s\),~\(t \in I\) and \(A\in \mA_{\mathrm{loc}}\) it holds that
    \begin{equation*}
        \partial_t\, \alpha_{s,t}\, A = \alpha_{s,t}\, \i\, \mL_{\Phi(t)}\, A .
    \end{equation*}
\end{definition}

\section{Results}
\label{sec:res}

From now on we fix a time-dependent zero-chain \(\Phi\) on an interval \(I \subseteq \R\) and decay functions \(F\),~\(G\) with \(\nu_F>2D+2\) and \(\nu_G > D+2\) (cf.~Definition~\ref{def:decay-function}).
We then assume that the terms of \(\Phi\) grow at most linearly.

\begin{assumption}
    \label{ass:growth}
    There is an \(x_0 \in \Gamma\) and a constant \(C_\Phi >0\), such that
    \begin{equation*}
        \sup_{t\in I} \, \norm{\Phi_x(t)}_{G,x}
        \leq
        C_\Phi \, \paren[\big]{1 + \dist{x,x_0}}
        \quadtext{for all}
        x \in \Gamma
        .
    \end{equation*}
\end{assumption}

Moreover, we assume a Lieb-Robinson bound with a linear light cone for all \emph{bounded} zero-chains that have the same decay as \(\Phi\).

\begin{assumption}
    \label{ass:LRB}
    There exist constants \(C_{\mathrm{LR}}\),~\(\vlr > 0\), such that for all time-dependent zero-chains \(\Psi \in \mathcal{Z}_{G,I}\) with associated cocycle of automorphisms \((\alpha_{s,t})_{ s,t \in I}\) and all \(A\in \alg_X\), \(B\in \alg_Y^+\), and \(s\),~\(t \in I\) it holds that
    \begin{equation*}
        \norm[\big]{\commutator[\big]{\alpha_{s,t} \, A ,\, B }}
        \leq
        C_{\mathrm{LR}} \, \norm{A} \, \norm{B} \, \abs{X}
        \, F\paren[\big]{
            \pospart[\big]{\dist{X,Y} - \vlr \, \tnorm{\Psi}_G\, \abs{t-s}}
        }
        ,
    \end{equation*}
    where \(\pospart{x}=x\) if \(x\geq 0\) and \(\pospart{x} = 0\) if \(x<0\).
\end{assumption}

This assumption is in particular satisfied for exponential localization, with the decay functions \(G(r)=\e^{-b r}\) and \(F(r)=\e^{-b' r}\) for some \(b>b'>0\).
To see this, one constructs an associated interaction, which is exponentially decaying in the sense of~\cite{NSY2017} with the function~\(r\mapsto F(r) \, (1+r)^{D+1+\epsi}\).
The result then follows by~\cite[Theorem~3.1]{NSY2017}.
Moreover, we expect it to be satisfied for polynomial localization as well.
Indeed, for spin systems with time-independent interactions and polynomial decay, linear light cones for large times have been shown in~\cite{KS2020}.
For Fermions, only Lieb-Robinson bounds with algebraic light cones, i.e.\ where the above bound holds with \(\dist{X,Y}\) replaced by \(\dist{X,Y}^\sigma\) for some \(\sigma\in \intervaloo{0,1}\), are known~\cite{TW2025}.
In this case, our results hold if one assumes \(
    \sup_{t\in I} \, \norm{\Phi_x(t)}_{G,x}
    \leq
    C_\Phi \, \paren[\big]{1 + \dist{x,x_0}}^\sigma
\) instead of \cref{ass:growth}.

We will apply Assumption~\ref{ass:LRB} to approximations of \(\Phi\) on finite subsets of~\(\Gamma\).
For this purpose, for each \(k\in \intervalco{0,\infty}\) we define the time-dependent zero-chain \(\Phi^k\) by
\begin{equation*}
    \Phi^{k}_x(t) \coloneq
    \begin{cases*}
        \E_{B_{k/2}(x)} \, \Phi_x(t) & for \(x\in B_{k/2}(x_0)\), and \\
        0                            & \text{otherwise.}
    \end{cases*}
\end{equation*}
It is defined such that \(\sum_{x\in \Gamma} \Phi^k_x\) is strictly localized in \(B_k(x_0)\) and \(\Phi^k \in \mathcal{Z}_{G,I}\) with \(\tnorm{\Phi^k}_G \leq C_\Phi\, (1+ \tfrac{k}{2})\).
We denote the cocycle generated by \(\Phi^k\) as \((\alpha^k_{s,t})_{s,t \in I}\).

We then obtain existence and uniqueness of the infinite volume dynamics for short times with an additional explicit Lieb-Robinson type estimate.

\begin{theorem}\label{thm:small-times}
    Let \(\tau \coloneq 1/\paren{4 \, \vlr\, C_\Phi}\).
    For all \(s\), \(t\in I\) with \(\abs{t-s} \leq \tau\) and \(A\in \alg\),
    \begin{equation*}
        \alpha_{s,t} \, A
        \coloneq
        \lim_{k\to\infty} \, \alpha^k_{s,t} \, A
    \end{equation*}
    exists in norm and the convergence is uniform in \(s\) and \(t\).
    Moreover, for subintervals \(I'\subset I\) with \(\abs{I'} \leq \tau\), \((\alpha_{s,t})_{s,t \in I'}\) is the unique cocycle of automorphisms generated by the restriction \(
        \Phi|_{I'\times \Gamma} \colon I' \times \Gamma \to \mA, \ (t,x)\mapsto \Phi_x(t)
    \) of the time-dependent zero-chain~\(\Phi\) to~\(I'\).

    Setting \(\mu \coloneq \min\paren[\big]{\nu_F-(2D+2), \nu_G-(D+2)}\), it holds that for each \(\nu \in \intervaloo{0,\mu}\), there exists \(\gamma_\nu > 0\), such that for all \(s\),~\(t \in I\) with \(\abs{t-s}\leq \tau\) and \(A\in \mA_\nu\) we have the bound
    \begin{equation}
        \label{eq:short-time-locality-of-automorphism-in-nu-norm}
        \norm{\alpha_{s,t} \, A }_{\nu,x_0}
        \leq
        \gamma_\nu \, \norm{A}_{\nu,x_0} .
    \end{equation}
    In particular, it holds that \(\alpha_{s,t} \, A \in \mA_\nu\).
    The constant \(\gamma_\nu\) does not depend on \(\Phi\).
\end{theorem}

While we excluded \(\nu=0\) in the bound~\eqref{eq:short-time-locality-of-automorphism-in-nu-norm}, from the convergence and properties of the automorphisms \(\alpha^k_{s,t}\), one immediately has \(\norm{\alpha_{s,t} \, A } \leq \norm{A}\) for all \(A\in \mA\) and \(s\),~\(t\in I\) with \(\abs{t-s} \leq \tau\).
Let us also stress that the bound~\eqref{eq:short-time-locality-of-automorphism-in-nu-norm} is influenced by \(C_\Phi\) as it only holds for \(\abs{t-s} \leq \tau = 1/(4 \, \vlr \, C_\Phi)\), even though~\(\gamma_\nu\) can be chosen uniformly for all~\(\Phi\) with the specified decay functions.

The idea of the proof is the following.
For the part of \(A\) that is localized in \(B_{k/8}(x_0)\), the restricted evolution \(\alpha_{s,s+\delta t}^k\) with Lieb-Robinson velocity \(v_{\mathrm{LR}} = c_\mathrm{LR} \, C_\Phi \, (1+\tfrac{k}{2})\) is a good approximation of \(\alpha_{s,s+\delta t}\) for large \(k\), as long as the enlarged support \(B_{k/8 + v_{\mathrm{LR}}\,\delta t}(x_0)\) is far from the boundary of \(B_{k/2}(x_0)\).
And this is the case for \(c_\mathrm{LR} \, C_\Phi \, \delta t \leq \tfrac{1}{4}\).
The actual proof is technically more difficult, because the Lieb-Robinson velocity only captures the growth of the support of most of the observable and one has to estimate the tails carefully.

The short-time result can then be extended by concatenation to existence and uniqueness for all times and a Lieb-Robinson type estimate with an exponential light cone.

\begin{theorem}
    \label{thm:all-times}
    The time-dependent zero-chain \(\Phi\) generates a unique cocycle of automorphisms \((\alpha_{s,t})_{s,t \in I}\).
    The cocycle \((\alpha^k_{s,t})_{s,t \in I}\) converges strongly to this cocycle as \(k \to \infty\), in the sense that for all \(A\in \mA\) and \(s\),~\(t \in I\) one has \(\alpha^k_{s,t}\, A \to \alpha_{s,t}\, A\) as \(k\to \infty\).

    Setting \(\mu \coloneq \min\paren[\big]{\nu_F-(2D+2), \nu_G-(D+2)}\), it holds that for each \(\nu \in \intervaloo{0,\mu}\), there are \(C_\nu > 0\), \(\gamma_\nu > 0\) that do not depend on \(\Phi\), such that for all \(s\),~\(t \in I\) and \(A\in \mA_\nu\) we have the bound
    \begin{equation}
        \label{eq:locality-of-automorphism-in-nu-norm}
        \norm{\alpha_{s,t} \, A }_{\nu,x_0}
        \leq
        C_\nu \, \e^{\gamma_\nu \, C_\Phi \, \abs{t-s}} \, \norm{A}_{\nu,x_0} .
    \end{equation}
    In particular, it holds that \(\alpha_{s,t}\, A \in \mA_\nu\).
\end{theorem}

As for the short-time result, we immediately obtain \(\norm{\alpha_{s,t} \, A} \leq \norm{A}\) for all \(A\in \alg\) and \(s\),~\(t\in I\).

Moreover, the bound~\eqref{eq:locality-of-automorphism-in-nu-norm} implies the usual commutator Lieb-Robinson bound with an exponential light cone:
For all \(k \geq0\) and \(Y\subset \Gamma\) with \(B_k(x_0)\cap Y = \emptyset\) and all \(A\in \alg_{B_k(x_0)}\), \(B\in \alg^+_Y\), denoting \(r=\dist[\big]{B_k(x_0), Y}\), one has for all \(0<c<1\)
\begin{align*}
    \norm{
        \commutator{\alpha_{s,t} \, A, B}
    }
    &\leq
    2 \, \norm{( 1- \E_{B_{k+cr}(x_0)}) \, \alpha_{s,t} \, A} \, \norm{B}
    + \norm{
        \commutator{
            \E_{B_{k+cr}(x_0)} \, \alpha_{s,t} \, A
            ,
            B
        }
    }
    \\&\leq
    2 \, (1+k+cr)^{-\nu} \, \norm{\alpha_{s,t} \, A}_{\nu,x_0} \, \norm{B}
    \\&\leq
    2 \, (1+k+cr)^{-\nu} \, (1+k)^\nu \, C_\nu\, \e^{\gamma_\nu \, C_\Phi \, \abs{t-s}} \, \norm{A} \, \norm{B}
    \\&\leq
    2 \, \norm{A} \, \norm{B} \, C_\nu \, \e^{\gamma_\nu \, C_\Phi \, \abs{t-s} - \nu \ln(1+cr/(1+k))}
\end{align*}
and therefore
\begin{equation}
    \label{eq:commutator-LRB}
    \norm{
        \commutator{\alpha_{s,t} \, A, B}
    }
    \leq
    2 \, \norm{A} \, \norm{B} \, C_\nu \, \e^{\gamma_\nu \, C_\Phi \, \abs{t-s} - \nu \ln(1+r/(1+k))}
    .
\end{equation}
This bound is referred to as a Lieb-Robinson bound with exponential light cone, since the right-hand side is small whenever
\begin{equation*}
    r
    \gg
    (1+k) \, \e^{C_\Phi \, \gamma_\nu \, \nu^{-1} \, \abs{t-s}}
    .
\end{equation*}

\section{Proofs}
\label{sec:proofs}

We provide the proof of \cref{thm:small-times,thm:all-times} in the following sections.
Some technical lemmas, which are necessary for the proofs, are given in \cref{app:technical-lemmas}.

\subsection{Existence and uniqueness for short times:\linebreak Proof of Theorem~\ref{thm:small-times}}
\label{sec:proof-short-times}

We first show that for every \(A\in\mA=\mA_0\) and all \(s\),~\(t\in I\) with \(\abs{t-s} \leq \tau\) the sequence \((\alpha^k_{s,t} \, A)_{k\in \N}\) is a Cauchy sequence in \(\mA\) with respect to the operator norm.
By completeness, it has a limit, which we denote \(\alpha_{s,t} \, A\).
And since this convergence is actually uniform in \(s\) and \(t\), we can later conclude that \(\alpha_{s,t}\) is the unique cocycle generated by \(\Phi\).
To prove the estimate~\eqref{eq:short-time-locality-of-automorphism-in-nu-norm} for \(\nu\in \intervaloo{0,\mu}\), we need a similar estimate for \(\norm{\alpha^l_{s,t} \, A - \alpha^k_{s,t} \, A}\) with explicit decay of the form \((1+k)^{-\nu}\) for all \(l\geq k\) and \(A\in \alg_{\nu}\).
To not do the same calculation twice, we treat all \(\nu\in[0,\mu)\) at once.

\paragraph{Cauchy type estimate}
Let \(\nu \in \intervalco{0,\mu}\), \(A \in \mA_\nu\) and \(s\),~\(t \in I\) with \(\abs{t-s} \leq \tau\), where, without loss of generality, we assume \(s \leq t\).
For any \(k\),~\(l \in \intervalco{0,\infty}\), with \(k \leq l\) we find
\begin{subequations}
    \begin{align}
        \Alignindent
        \norm{\alpha^l_{s,t}\, A - \alpha^k_{s,t}\, A}
        \nonumber
        \\&\leq
        \norm{(\alpha^l_{s,t} - \alpha^k_{s,t}) \, ( 1- \E_{B_{k/8}(x_0)})\, A}
        + \norm{(\alpha^l_{s,t} - \alpha^k_{s,t}) \, \E_{B_{k/8}(x_0)}\, A}
        \nonumber
        \\&\leq
        2\, \norm{( 1- \E_{B_{k/8}(x_0)})\, A}
        + \int_s^t \d u \, \norm{\partial_u\, \alpha^l_{s,u}\,\alpha^k_{u,t}\, \E_{B_{k/8}(x_0)} \, A}
        \nonumber
        \\&\leq
        2\, \norm{( 1- \E_{B_{k/8}(x_0)})\, A}
        + \int_s^t \d u \, \sumstack[lr]{x \in B_{l/2}(x_0)} \, \norm{\commutator{\Phi^l_x(u) - \Phi^k_x(u) , \alpha^k_{u,t}\, \E_{B_{k/8}(x_0)} \, A}}
        \nonumber
        \\&\leq
        2\, \norm{( 1- \E_{B_{k/8}(x_0)})\, A}
        \label{eq:splitting-locality-A}
        \\&\relskip + \int_s^t \d u \qquad \sumstack[lr]{x\in B_{k/2}(x_0)} \quad \norm{\commutator{(\E_{B_{l/2}(x)} - \E_{B_{k/2}(x)}) \,\Phi_x(u) , \alpha^k_{u,t}\, \E_{B_{k/8}(x_0)} \, A}}
        \label{eq:splitting-locality-interaction-x-inside}
        \\&\relskip + \int_s^t \d u \qquad \sumstack[lr]{x\in B_{l/2}(x_0) \setminus B_{k/2}(x_0)} \quad \norm{\commutator{\E_{B_{l/2}(x)} \, ( 1 - \E_{B_{\dist{x,x_0}/4}(x)} ) \, \Phi_x(u) , \alpha^k_{u,t}\, \E_{B_{k/8}(x_0)}\, A}}
        \label{eq:splitting-locality-interaction-x-outside-i}
        \\&\relskip + \int_s^t \d u \qquad \sumstack[lr]{x\in B_{l/2}(x_0) \setminus B_{k/2}(x_0)} \quad \norm{\commutator{\E_{B_{l/2}(x)} \, \E_{B_{\dist{x,x_0}/4}(x)} \, \Phi_x(u) , \alpha^k_{u,t}\, \E_{B_{k/8}(x_0)}\, A}}
        .
        \label{eq:splitting-locality-interaction-x-outside-ii}
    \end{align}
\end{subequations}
We bound each of the four terms separately.
For \(\nu=0\), the term~\eqref{eq:splitting-locality-A} converges to~\(0\) because \(A\) is quasi-local, as explained before \cref{def:decay-function}.
For \(\nu>0\), we have \(A\in \alg_\nu\) and thus~\eqref{eq:splitting-locality-A} is bounded by
\begin{equation*}
    \text{\eqref{eq:splitting-locality-A}}
    \leq
    2 \, \frac{1}{(1+\tfrac{k}{8})^\nu} \, \norm{A}_{\nu,x_0}
    \leq
    2\, \frac{8^\nu }{(1+k)^\nu} \, \norm{A}_{\nu,x_0}
    .
\end{equation*}
The remaining estimates all work for \(\nu \geq 0\).
The second and third terms are bounded using only the decay of the quasi-local terms of \(\Phi\).
In both cases we use the trivial bound for the commutator.
The term~\eqref{eq:splitting-locality-interaction-x-inside} is bounded by
\begin{align*}
    \text{\eqref{eq:splitting-locality-interaction-x-inside}}
    &\leq
    2 \, \tau \sup_{u\in I} \sum_{x\in B_{k/2}(x_0)} \norm{\E_{B_{l/2}(x)} \, (1 - \E_{B_{k/2}(x)}) \,\Phi_x(u)} \, \norm{A}
    \\&\leq
    2 \, \tau \sup_{u\in I} \sum_{x\in B_{k/2}(x_0)} \norm{\Phi_x(u)}_{G,x} \, G(k/2) \, \norm{A}
    \\&\leq
    2 \, \tau \sum_{x\in B_{k/2}(x_0)} C_\Phi \, \paren[\big]{1+ \dist{x,x_0}} \, G(k/2) \, \norm{A}
    \\&\leq
    2 \, \tau \, C_{\mathrm{vol}} \, C_\Phi \, \frac{(1+k/2)^{D+1+\nu+\eps}\, G(k/2)}{(1+k/2)^{\nu+\eps}} \, \norm{A}
    \\&\leq
    \frac{C_{\mathrm{vol}}}{2 \, \vlr} \, \frac{2^{\nu+\eps}\, C}{(1+k)^{\nu+\eps}}\, \norm{A}
    ,
\end{align*}
for some \(C>0\) and an \(\eps>0\), such that \(D+1+\nu+\eps < \nu_G\).
Here we used that \(D+1+\nu+\eps < \nu_G\) and therefore \(k \mapsto (1+\frac{k}{2})^{D+1+\nu+\eps}\, G(\frac{k}{2})\) is bounded.
For~\eqref{eq:splitting-locality-interaction-x-outside-i} we apply the same bounds to the commutator and then use the decay of \(G\) together with the volume-growth assumption to treat the infinite sum and obtain the upper bound
\begin{align*}
    \text{\eqref{eq:splitting-locality-interaction-x-outside-i}}
    &\leq
    2 \, \tau \, \sup_{u\in I}\, \sum_{x\in \Gamma \setminus B_{k/2}(x_0)} \norm{( 1 - \E_{B_{\dist{x,x_0}/4}(x)} ) \, \Phi_x(u)} \, \norm{A}
    \\&\leq
    2 \, \tau \sum_{x\in \Gamma \setminus B_{k/2}(x_0)} C_\Phi \, \paren[\big]{1+\dist{x,x_0}} \, G\paren[\big]{\dist{x,x_0}/4} \, \norm{A}
    \\&\leq
    2 \, \tau \, C_\Phi \, \sup_{m \geq k/2 } \, (1 + m)^{D+2+\eps} \, G(m/4) \, \sum_{x \in \Gamma } \frac{1}{\paren[\big]{1 + \dist{x,x_0}}^{D+1+\eps}} \, \norm{A}
    \\&\leq
    \frac{1}{2 \, \vlr} \, \frac{C}{(1+k)^{\nu+\eps}} \, \norm{A}
    ,
\end{align*}
for some \(C>0\) and an \(\eps>0\), such that \(D+2+\nu+2\eps < \nu_G\).
The last sum converges due to Lemma~\ref{lem:summability}, and we used that the map \(m \mapsto (1+m)^{D+2+\nu+2\eps}\, G(\frac{m}{4})\) is bounded.
To bound~\eqref{eq:splitting-locality-interaction-x-outside-ii} we apply the Lieb-Robinson bound from \cref{ass:LRB} for the cocycle of automorphisms \(\alpha^k\), which is generated by the uniformly \(G\)-localized time-dependent zero-chain \(\Phi^k\).
For this, we first note that
\begin{align*}
    \Alignindent
    \dist[\big]{B_{\dist{x,x_0}/4}(x), B_{k/8}(x_0)}
    - \vlr \, \tnorm{\Phi^k}_G \, \abs{t-s}
    \\&\geq
    \tfrac{3}{4} \, \dist{x,x_0}-k/8
    - \vlr \, C_\Phi \, (1+k/2) \, \tau
    \\&\geq
    \tfrac{3}{4} \, \dist{x,x_0}-k/8
    - (1+k/2)/4
    \\&\geq
    \tfrac{3}{4} \, \dist{x,x_0} - (k+1)/4
    .
\end{align*}
Then, \eqref{eq:splitting-locality-interaction-x-outside-ii} is bounded by
\begin{align*}
    \text{\eqref{eq:splitting-locality-interaction-x-outside-ii}}
    &\leq
    \tau \, C_{\mathrm{LR}} \, \norm{A} \sumstack{x\in \Gamma \setminus B_{k/2}(x_0)} \sup_{u\in I} \, \norm{\Phi_x(u)} \, \abs{B_{k/8}(x_0)}
    \, F\paren[\big]{
        \pospart[\big]{\tfrac{3}{4} \, \dist{x,x_0} - \tfrac{k+1}{4}}
    }
    \\&\leq
    \tau \, C_{\mathrm{LR}} \, \norm{A} \sumstack[r]{x\in \Gamma \setminus B_{k/2}(x_0)} \, C_\Phi \, \paren[\big]{1+\dist{x,x_0}} \, C_{\mathrm{vol}} \, (1+k/8)^D
    \, F\paren[\big]{
        \pospart[\big]{\tfrac{3}{4} \, \dist{x,x_0} - \tfrac{k+1}{4}}
    }
    \\&\leq
    \tau \, C_{\mathrm{LR}} \, C_\Phi \, C_{\mathrm{vol}} \, \norm{A} \sumstack[r]{x\in \Gamma} \, \paren[\big]{1+\dist{x,x_0}}^{-(D+1+\epsi)}
    \, \sup_{\mathclap{m\geq k/2}} \, (1+m)^{2D+2+\epsi} \, F\paren[\big]{
        \pospart[\big]{\tfrac{3}{4} \, m - \tfrac{k+1}{4}}
    }
    \\&\leq
    \frac{C_{\mathrm{LR}} \, C_\mathrm{vol}}{4 \, \vlr} \, \frac{C}{{(1+k)^{\nu+\eps}}} \, \norm{A}
    ,
\end{align*}
for some \(C>0\) and an \(\eps>0\), such that \(2D+2+\nu+2\eps<\nu_F\).
This time we used that the map
\(
    k
    \mapsto \sup_{m\geq k/2} \, (1+m)^{2D+2+\nu+2\eps} \, F\paren[\big]{
        \pospart[\big]{\tfrac{3}{4} \, m - \tfrac{k+1}{4}}
    }
\)
is bounded, which we show in Lemma~\ref{lem:sup-function-bound}.

Combining the four bounds for \(\nu=0\) we have shown that for all \(A\in \mA\)
\begin{equation*}
    \norm{\alpha^l_{s,t}\, A - \alpha^k_{s,t}\, A}
    \to
    0
    \quad\text{uniformly for all \(s\),~\(t\in I\) with \(\abs{t-s} \leq \tau\)}.
\end{equation*}
And for \(\nu\in \intervaloo{0,\mu}\), we have shown that there is a constant \(\tilde{\gamma}_\nu>0\), that does not depend on \(\Phi\), such that for all \(A\in \alg_\nu\) and \(s\), \(t\in I\) satisfying \(\abs{t-s} \leq \tau = 1/(4 \, \vlr \, C_\Phi)\) and all \(l \geq k \in \intervalco{0,\infty}\) it holds that
\begin{equation}
    \label{eq:cauchy-bound}
    \norm{\alpha^l_{s,t}\, A - \alpha^k_{s,t}\, A}
    \leq
    \frac{\tilde{\gamma}_\nu}{(1+k)^\nu} \, \norm{A}_{\nu,x_0}
    .
\end{equation}

\paragraph{Convergence}
By the Cauchy estimate for \(\nu=0\), the sequence \((\alpha^k_{s,t}\, A)_{k\in \N_0}\) converges for all \(A\in \alg\), and we denote its limit by \(\alpha_{s,t}\, A\).
Moreover, this convergence is uniform for all \(s\),~\(t\in I\) with \(\abs{t-s}\leq\tau\).

\paragraph{Cocycle and generator properties}
Let \(I'\subset I\) be a subinterval with \(\abs{I'} \leq \tau\).
It is easy to see that \((\alpha_{s,t})_{s,t \in I'}\) is a strongly continuous cocycle of automorphisms on \(\mA\), since it inherits all the relevant properties from the approximations \((\alpha^k_{s,t})_{s,t \in I'}\).

To show that this cocycle is generated by the time-dependent zero-chain \(\Phi|_{I'\times \Gamma}\), let \(s\),~\(t\in I'\).
Then, note that for all \(h \in \R\) such that \(t+h \in I'\) and all \(A \in \mA_{\mathrm{loc}}\), due to the strong continuity of \((\alpha^k_{s,t})_{s,t \in I'}\) and continuity of \(u \mapsto \mL_{\Phi^k(u)} \, A\), it holds that
\begin{equation*}
    \alpha^k_{s,t+h}\, A - \alpha^k_{s,t}\, A
    =
    \int_{t}^{t+h} \d u \, \alpha^k_{s,u}\, \i \, \mL_{\Phi^k(u)} \, A
    .
\end{equation*}
Together with the uniform convergence and Lemma~\ref{lem:convergence-of-liouvillian} this gives us in the limit \(k \to \infty\) that
\begin{equation*}
    \alpha_{s,t+h}\, A - \alpha_{s,t}\, A
    =
    \int_{t}^{t+h} \d u \, \alpha_{s,u}\, \i \, \mL_{\Phi(u)} \, A
    .
\end{equation*}
By the strong continuity of \((\alpha_{s,t})_{s,t \in I'}\) and Lemma~\ref{lem:continuity-of-liouvillian}, it follows that
\begin{equation*}
    \partial_t \, \alpha_{s,t} \, A
    =
    \alpha_{s,t} \, \i \, \mL_{\Phi(t)} \, A
    .
\end{equation*}

\paragraph{Uniqueness} To show uniqueness, let \((\tilde{\alpha}_{s,t})_{s,t \in I'}\) be any cocycle of automorphisms, generated by~\(\Phi|_{I'\times \Gamma}\).
We can show that \((\alpha^k_{s,t})_{s,t \in I'}\) also converges strongly to it, thereby showing that it must be identical to \((\alpha_{s,t})_{s,t \in I'}\).
For this, let \(k\in \intervalco{0,\infty}\), \(A\in \mA\) and bound \(\norm{\tilde{\alpha}_{s,t}\, A - \alpha^k_{s,t}\, A}\) exactly as we did to arrive at the terms~\eqref{eq:splitting-locality-A}–\eqref{eq:splitting-locality-interaction-x-outside-ii}.
This results in
\begin{subequations}
    \begin{align*}
        \Alignindent
        \norm{\tilde{\alpha}_{s,t}\, A - \alpha^k_{s,t}\, A}
        \\&\leq
        \norm{(\tilde{\alpha}_{s,t} - \alpha^k_{s,t}) \, ( 1- \E_{B_{k/8}(x_0)})\, A}
        + \norm{(\tilde{\alpha}_{s,t} - \alpha^k_{s,t}) \, \E_{B_{k/8}(x_0)}\, A}
        \\&\leq
        2\, \norm{( 1- \E_{B_{k/8}(x_0)})\, A}
        + \int_s^t \d u \, \norm{\partial_u\, \tilde{\alpha}_{s,u}\,\alpha^k_{u,t}\, \E_{B_{k/8}(x_0)} \, A}
        \\&\leq
        2\, \norm{( 1- \E_{B_{k/8}(x_0)})\, A}
        + \int_s^t \d u \, \sumstack[lr]{x \in \Gamma} \, \norm{\commutator{\Phi_x(u) - \Phi^k_x(u) , \alpha^k_{u,t}\, \E_{B_{k/8}(x_0)} \, A}}
        \\&\leq
        2\, \norm{( 1- \E_{B_{k/8}(x_0)})\, A}
        \\&\relskip + \int_s^t \d u \qquad \sumstack[lr]{x\in B_{k/2}(x_0)} \quad \norm{\commutator{(1 - \E_{B_{k/2}(x)}) \,\Phi_x(u) , \alpha^k_{u,t}\, \E_{B_{k/8}(x_0)} \, A}}
        \\&\relskip + \int_s^t \d u \qquad \sumstack[lr]{x\in \Gamma \setminus B_{k/2}(x_0)} \quad \norm{\commutator{( 1 - \E_{B_{\dist{x,x_0}/4}(x)} ) \, \Phi_x(u) , \alpha^k_{u,t}\, \E_{B_{k/8}(x_0)}\, A}}
        \\&\relskip + \int_s^t \d u \qquad \sumstack[lr]{x\in \Gamma \setminus B_{k/2}(x_0)} \quad \norm{\commutator{\E_{B_{\dist{x,x_0}/4}(x)} \, \Phi_x(u) , \alpha^k_{u,t}\, \E_{B_{k/8}(x_0)}\, A}}
        .
    \end{align*}
\end{subequations}
These four terms can be bounded by the exact same steps used to bound the previous four terms, thus showing that \(\norm{\tilde{\alpha}_{s,t}\, A - \alpha^k_{s,t}\, A} \to 0\) as \(k \to \infty\).
Hence, \(\tilde{\alpha}_{s,t} \, A = \alpha_{s,t} \, A\) for all \(s\),~\(t\in I'\).

\paragraph{Growth estimate}
Clearly for all \(A\in \alg_\nu\) and \(s\),~\(t\in I\) with \(\abs{t-s} \leq \tau\), it holds that \(\norm{\alpha_{s,t} \, A} \leq \norm{A}\).
And to prove~\eqref{eq:short-time-locality-of-automorphism-in-nu-norm} it is left to estimate the locality of \(\alpha_{s,t} \, A\).
For this, we bound
\begin{align*}
    \Alignindent
    \norm{
        (1-\E_{B_k(x_0)})
        \, \alpha_{s,t}
        \, A
    }
    \\&=
    \norm{
        (1-\E_{B_k(x_0)})
        \, (
        \alpha_{s,t}
        - \alpha^k_{s,t} \, \E_{B_k(x_0)}
        )
        \, A
    }
    \\&\leq
    \norm{
        (1-\E_{B_k(x_0)})
        \, (\alpha_{s,t} - \alpha^k_{s,t})
        \, A
    }
    + \norm{
        (1-\E_{B_k(x_0)})
        \, \alpha^k_{s,t}
        \, (1 - \E_{B_k(x_0)})
        \, A
    }
    \\&\leq
    2 \, \norm{
        (\alpha_{s,t} - \alpha^k_{s,t}) \, A
    }
    + \frac{2 \, \norm{ A }_{\nu,x_0}}{(1+k)^\nu}
    \\&\leq
    2 \, \frac{\tilde{\gamma}_\nu \, \norm{ A}_{\nu,x_0}}{(1+k)^\nu}
    + \frac{2\,\norm{ A }_{\nu,x_0}}{(1+k)^\nu}
    ,
\end{align*}
for all \(k\in \intervalco{0,\infty}\) using~\eqref{eq:cauchy-bound} and locality of~\(A\).
This proves that \(\alpha_{s,t} \, \alg_\nu \subset \alg_\nu\) and
\begin{equation*}
    \norm{\alpha_{s,t}\, A}_{\nu,x_0}
    \leq
    3 \, (\tilde{\gamma}_\nu + 1) \, \norm{ A}_{\nu,x_0}.
\end{equation*}

\subsection{Existence and uniqueness for all times:\linebreak Proof of Theorem~\ref{thm:all-times}}
\label{sec:proof-all-times}

Next, we prove \cref{thm:all-times} by lifting the results from \cref{thm:small-times} to all times.

\paragraph{Convergence and cocycle and generator properties}
Let \(\tau = 1/\paren{4\,\vlr\,C_\Phi}\), let \(s\),~\(t \in I\) and \(A \in \mA\).
Without loss of generality we assume that \(s\leq t\).
We choose an \(N\in \N_0\) and an increasing tuple \((t_i)_{i\in \List{0,\dotsc,N}}\) of elements of~\(I\), such that \(t_0=s\), \(t_N=t\), and \(t_{i+1}-t_i \leq \tau\).
We know by Theorem~\ref{thm:small-times} that for all \(i \in \List{0,\dotsc,N-1}\) the restriction \(\Phi|_{\intervalcc{t_i , t_{i+1}}\times \Gamma}\) of \(\Phi\) to \(\intervalcc{t_i , t_{i+1}}\) generates a unique cocycle of automorphisms \((\alpha_{s,t})_{s,t\in \intervalcc{t_i , t_{i+1}}}\) that can be approximated in a strong sense by the cocycle \((\alpha^k_{s,t})_{s,t\in \intervalcc{t_i , t_{i+1}}}\).
It holds that
\begin{align*}
    \Alignindent
    \norm[\bigg]{
        \alpha_{s,t}^k \, A
        - \paren[\bigg]{\prod_{i=0}^{N-1} \alpha_{t_i,t_{i+1}}} \, A
    }
    \\&=
    \norm[\bigg]{
        \paren[\bigg]{\prod_{i=0}^{N-1} \alpha_{t_i,t_{i+1}}^k} \, A
        - \paren[\bigg]{\prod_{i=0}^{N-1} \alpha_{t_i,t_{i+1}}} \, A
    }
    \\&\leq
    \sum_{j=0}^{N-1} \, \norm[\bigg]{
        \paren[\bigg]{\prod_{i=0}^{j-1} \alpha_{t_i,t_{i+1}}^k}
        \, \paren[\big]{\alpha_{t_j,t_{j+1}}^k - \alpha_{t_j,t_{j+1}}}
        \, \paren[\bigg]{\prod_{i=j+1}^{N-1} \alpha_{t_i,t_{i+1}}} \, A
    }
    \\&\leq
    \sum_{j=0}^{N-1} \, \norm[\bigg]{
        \paren[\big]{\alpha_{t_j,t_{j+1}}^k - \alpha_{t_j,t_{j+1}}}
        \, \paren[\bigg]{\prod_{i=j+1}^{N-1} \alpha_{t_i,t_{i+1}}} \, A
    }
    \\&\to 0 \quadtext{as} k\to\infty
    ,
\end{align*}
because \(\paren[\big]{\prod_{i=j+1}^{N-1} \alpha_{t_i,t_{i+1}}} \, A\) is a fixed element of \(\mA\) and \(\alpha_{t_j,t_{j+1}}^k \to \alpha_{t_j,t_{j+1}}\) strongly on~\(\mA\), by \cref{thm:small-times}.
Therefore, \(\alpha_{s,t}^k \, A\) converges to \(
    \alpha_{s,t}\, A
    \coloneq
    \paren[\big]{\prod_{i=0}^{N-1} \alpha_{t_i,t_{i+1}}} \, A
\).
In particular, \(\alpha_{s,t}\, A\) is independent of the choice of intermediate times \((t_i)_{i\in \List{0,\dotsc,N}}\).
It is now easy to see that \((\alpha_{s,t})_{s,t\in I}\) defines a cocycle of automorphisms and is generated by \(\Phi\), since it inherits all the relevant properties from the short time cocycles \((\alpha_{s,t})_{s,t\in \intervalcc{t_i , t_{i+1}}}\).

\paragraph{Uniqueness}
For the uniqueness, let \((\tilde{\alpha}_{s,t})_{s,t\in I}\) be any cocycle generated by \(\Phi\).
We can split it up in the same way as above
\begin{equation*}
    \tilde{\alpha}_{s,t} \, A
    =
    \paren[\bigg]{\prod_{i=0}^{N-1} \tilde{\alpha}_{t_i,t_{i+1}}} \, A.
\end{equation*}
We observe that for each \(i \in \List{0,\dotsc,N-1}\) the cocycle \((\tilde{\alpha}_{s,t})_{s,t\in \intervalcc{t_i,t_{i+1}}}\) is generated by \(\Phi|_{\intervalcc{t_i,t_{i+1}}\times \Gamma}\).
Together with the uniqueness statement of Theorem~\ref{thm:small-times}, this lets us conclude that \((\tilde{\alpha}_{s,t})_{s,t\in I} = (\alpha_{s,t})_{s,t\in I}\).

\paragraph{Growth estimate}
To obtain the bound, let \(A\in \mA_\nu\) for some \(\nu \in \intervaloo{0,\mu}\) and choose the tuple from above as \(t_i = s + i\, \tau\) for \(i \in \List{0, \dots , \floor[\big]{\frac{t-s}{\tau}}}\) and \(t_{\floor[\big]{\frac{t-s}{\tau}}+1} = t\).
With the bound of Theorem~\ref{thm:small-times} we now find
\begin{equation*}
    \norm{\alpha_{s,t} \, A}_{\nu,x_0}
    =
    \norm[\bigg]{\prod_{i=0}^{\floor[\big]{\frac{t-s}{\tau}}} \alpha_{t_i,t_{i+1}} \, A}_{\nu,x_0}
    \leq
    \gamma_\nu^{\floor[\big]{\frac{t-s}{\tau}}+1} \, \norm{A}_{\nu, x_0}
    \leq
    \gamma_\nu\, \e^{\frac{t-s}{\tau} \ln(\gamma_\nu)}\, \norm{A}_{\nu, x_0}
    .
\end{equation*}
Replacing \(1/\tau = 4 \, \vlr \, C_\Phi\) and recalling that \(\gamma_\nu\) is independent of \(\Phi\), we obtain the estimate~\eqref{eq:locality-of-automorphism-in-nu-norm}.

\appendix
\crefalias{section}{appendix}

\section{Technical lemmas}
\label{app:technical-lemmas}

In this section we refer to the decay functions \(F\) and \(G\) and the time-dependent zero-chains \(\Phi\) and \(\Phi^k\) as defined in Section~\ref{sec:res}.

\begin{lemma}\label{lem:summability}
    For all \(\eps>0\) the sum
    \begin{equation*}
        \sum_{x \in \Gamma } \frac{1}{\paren[\big]{1 + \dist{x,x_0}}^{D+1+\eps}}
    \end{equation*}
    converges absolutely.
\end{lemma}
\begin{proof}
    For \(k \in \N\) we define \(S_k \coloneq B_k(x_0)\setminus B_{k-1}(x_0)\) and \(S_0= B_0(x_0)\).
    Due to the volume growth property of \((\Gamma,d)\) is holds that \(\abs{S_k} \leq C_{\mathrm{vol}} \, (1+k)^D\).
    From this we conclude
    \begin{align*}
        \sum_{x \in \Gamma } \frac{1}{\paren[\big]{1 + \dist{x,x_0}}^{D+1+\eps}}
        &\leq
        \sum_{k=1}^\infty \sum_{x \in S_k } \frac{1}{k^{D+1+\eps}} + S_0
        \leq
        \sum_{k=1}^\infty \frac{C_{\mathrm{vol}} \, (1+k)^D}{k^{D+1+\eps}} + C_{\mathrm{vol}}
        \\&\leq
        \sum_{k=1}^\infty \frac{C_{\mathrm{vol}} \, 2^D\, }{k^{1+\eps}} + C_{\mathrm{vol}}
        <
        \infty
        .\qedhere
    \end{align*}
\end{proof}

\begin{lemma}\label{lem:sup-function-bound}
    Let \(0 \leq \nu< \nu_F\).
    It holds that
    \begin{equation*}
        \sup_{k\geq 0} \, \sup_{m \geq k/2} \, (1+m)^\nu \, F\paren[\bigg]{\pospart[\bigg]{\frac{3m}{4} - \frac{k+1}{4}}}
        <
        \infty .
    \end{equation*}
\end{lemma}
\begin{proof}
    Since the expression is bounded away from \(\infty\) it is sufficient to consider the supremum for \(k \geq 2\) and bound
    \begin{align*}
        \Alignindent
        \sup_{k\geq 2} \, \sup_{m \geq k/2} \, (1+m)^\nu \, F\paren[\bigg]{\frac{3m}{4} - \frac{k+1}{4}}
        \\&=
        \sup_{k\geq 2} \, \sup_{m \geq (k/6- 1/3)} \, \paren[\bigg]{1+m+\frac{k+1}{3}}^\nu \, F\paren[\bigg]{\frac{3m}{4}}
        \\&=
        \sup_{k\geq 0} \, \sup_{m \geq k/6} \, \paren[\bigg]{1+m+\frac{k+3}{3}}^\nu \, F\paren[\bigg]{\frac{3m}{4}}
        \\&\leq
        \sup_{k\geq 0} \, \sup_{m \geq k/6} \, (2+3m)^\nu \, F\paren[\bigg]{\frac{3m}{4}}
        \\&\leq
        4^\nu \, \sup_{m \geq 0} \, \paren[\bigg]{1+\frac{3m}{4}}^\nu \, F\paren[\bigg]{\frac{3m}{4}}
        \\&<
        \infty
        ,
    \end{align*}
    where we substituted \(m \to m + (k+1)/3\) and \(k \to k+2\) in the second and third step, respectively.
\end{proof}

\begin{lemma}\label{lem:convergence-of-liouvillian}
    For all \(A\in \mA_{\mathrm{loc}}\) it holds that
    \begin{equation*}
        \sup_{t\in I} \, \norm{\mL_{\Phi(t)}\, A - \mL_{\Phi^k(t)}\, A} \to 0
        \quadtext{as}
        k\to \infty
        .
    \end{equation*}
\end{lemma}
\begin{proof}
    Let \(A\in \mA_{\mathrm{loc}}\) and \(k\in \N_0\). We have
    \begin{align*}
        \Alignindent
        \sup_{t\in I} \, \norm{\mL_{\Phi(t)}\, A - \mL_{\Phi^k(t)}\, A}
        \\&\leq
        \sup_{t\in I} \, \sumstack[r]{x\in B_{k/2}(x_0)} \, \norm[\big]{\commutator[\big]{(1-\E_{B_k(x)}) \, \Phi_x(t), A}}
        + \sup_{t\in I} \, \sumstack[r]{x\in \Gamma \setminus B_{k/2}(x_0)} \, \norm[\big]{\commutator[\big]{\Phi_x(t), A}}
        .
    \end{align*}
    The first term is bounded by
    \begin{equation*}
        \sum_{x\in B_{k/2}(x_0)} \sup_{t\in I} \, 2 \, \norm{\Phi_x(t)}_{G,x} \, G(k)\, \norm{A}
        \leq
        2 \, C_{\mathrm{vol}} \, \paren[\big]{1+\tfrac{k}{2}}^{D}\, C_\Phi \, \paren[\big]{1+\tfrac{k}{2}} \, G(k)\, \norm{A}
        ,
    \end{equation*}
    and since \(D+1 < \nu_G\) this bound vanishes as \(k\to \infty\).
    For the second term we assume that \(k\) is large enough so that \(A\) is supported in \(B_{k/8}(x_0)\).
    This allows us to insert a conditional expectation and then bound the second term by
    \begin{align*}
        \Alignindent
        \sup_{t\in I} \, \sum_{x\in \Gamma \setminus B_{k/2}(x_0)} \norm[\big]{\commutator[\big]{(1-\E_{B_{\dist{x,x_0}/4}(x)})\,\Phi_x(t), A}}
        \\&\leq
        \sum_{x\in \Gamma \setminus B_{k/2}(x_0)} \sup_{t\in I} \, 2 \, \norm{\Phi_x(t)}_{G,x} \, G\paren[\big]{\dist{x,x_0}/4} \, \norm{A}
        \\&\leq
        2 \, \sumstack[lr]{x\in \Gamma \setminus B_{k/2}(x_0)} \, C_\Phi\, \paren[\big]{1+\dist{x,x_0}} \, G\paren[\big]{\dist{x,x_0}/4} \, \norm{A}
        \\&\leq
        \frac{2 \, C_\Phi}{(1+k/2)^\epsi} \, \sup_{m \geq k/2} \, (1+m)^{D+2+2\eps}\, G(m/4) \, \sum_{x\in \Gamma} \frac{1}{\paren[\big]{1+\dist{x,x_0}}^{D+1+\eps}} \, \norm{A}
        ,
    \end{align*}
    for an \(\eps >0\), such that \(D+2+2\eps < \nu_G\).
    The final sum converges as shown in Lemma~\ref{lem:summability}, the supremum is bounded and hence the expression converges to \(0\) as \(k\to \infty\).
\end{proof}

\begin{lemma}\label{lem:continuity-of-liouvillian}
    For all \(A\in \mA_{\mathrm{loc}}\) the map \(I \to \mA, \, t \mapsto \mL_{\Phi(t)}\, A\) is continuous.
\end{lemma}
\begin{proof}
    For all \(k \in \N_0\) and all \(A\in \mA_{\mathrm{loc}}\), we know that the map \(I \to \mA, \, t \mapsto \mL_{\Phi^k(t)}\, A\) is continuous, since each of the finitely many quasi-local terms is continuous.
    Together with the uniform convergence of Lemma~\ref{lem:convergence-of-liouvillian}, this implies the claim.
\end{proof}

\statement{Acknowledgments}
This work was supported by the Deutsche Forschungsgemeinschaft (DFG, German Research Foundation) through TRR~352 (470903074) and FOR~5413 (465199066).

\printbibliography[heading=bibintoc]

\end{document}